\documentclass[prl,twocolumn,showpacs,superscriptaddress,amsmath,amssymb]{revtex4-2}
\usepackage{graphicx}
\usepackage{dcolumn}
\usepackage{bm}
\usepackage{hyperref}
\usepackage{braket}
\usepackage{xcolor}
\definecolor{darkgreen}{rgb}{0.0, 0.5, 0.13}
\usepackage{hyperref}
\usepackage{cleveref}
\usepackage{lipsum}
\usepackage{enumitem}
\newcommand{\subscript}[2]{$#1 _ #2$}
\usepackage{subfig}
\usepackage{graphicx}
\usepackage{tikz}

\DeclareMathOperator{\Tr}{Tr}

\DeclareMathOperator*{\argmin}{arg\,min}
\newcommand{\SV}{\mathcal{S}}

\newcommand{\energyBound}{{\bm{\omega}}}

\newcommand{\classicalSet}[1]{\mathcal{C}_{#1}}
\newcommand{\quantumSet}[1]{\mathcal{Q}_{#1}}
\newcommand{\mdlIneq}[1]{I_{#1,\delta}}
\newcommand{\prepBox}{P}
\newcommand{\msrmntBox}{M}
\newcommand{\prob}{{\bf p}}
\newcommand{\Hmin}{H_{\rm min}}
\newcommand{\N}{\mathcal{N}}
\renewcommand{\P}{\mathcal{P}}

\newcommand{\C}{\mathcal{C}}
\newcommand{\R}{\mathbb R}
\newcommand{\p}{\bm{p}}
\newcommand{\Ci}{C_i}

\newcommand{\CC}{\bm{\mathrm{C}}}

\newcommand{\ptr}[2]{\mathrm{Tr}_{#1}\left[#2\right]}
\newcommand{\id}{\mathcal{I}}

\newcommand{\epsmooth}{\epsilon_{s}}

\newcommand{\pr}[1]{\mathrm{Pr}\left[#1\right]}
\newcommand{\errV}{\epsilon_V}
\newcommand{\errK}{\epsilon_K}
\newcommand{\errW}{\epsilon_\Omega}
\newcommand{\Max}[1]{\mathrm{Max}[#1]}
\newcommand{\Min}[1]{\mathrm{Min}[#1]}

\newcommand{\expo}{\mathrm{e}}

\usepackage{amsthm}
\newtheorem{definition}{Definition}
\newtheorem{theorem}{Theorem}
\newtheorem{lemma}{Lemma}
\crefname{lemma}{Lemma}{Lemmas}
\newtheorem{proof-sketch}{Proof Sketch}
\usepackage{algorithmicx}
\usepackage[Algorithm,ruled]{algorithm} 
\usepackage{algorithm}
\usepackage{algcompatible}

\begin{document}

\title{Semi-device-independent full randomness amplification based on energy bounds}
\author{Gabriel Senno}
\affiliation{ICFO-Institut de Ciencies Fotoniques, The Barcelona
Institute of Science and Technology, 08860 Castelldefels,
Barcelona, Spain}
\author{Antonio Ac\'in}
\affiliation{ICFO-Institut de Ciencies Fotoniques, The Barcelona
Institute of Science and Technology, 08860 Castelldefels,
Barcelona, Spain}
\affiliation{ICREA, Pg. Lluis Companys 23, 08010 Barcelona, Spain}

\begin{abstract}
Quantum Bell nonlocality allows for the design of protocols that amplify the randomness of public and arbitrarily biased Santha-Vazirani sources, a classically impossible task. Information-theoretical security in these protocols is certified in a device-independent manner, i.e. solely from the observed nonlocal statistics and without any assumption about the inner-workings of the intervening devices. On the other hand, if one is willing to trust on a complete quantum-mechanical description of a protocol's devices, the elementary scheme in which a qubit is alternatively measured in a pair of mutually unbiased bases is, straightforwardly, a protocol for randomness amplification. In this work, we study the unexplored middle ground. We prove that full randomness amplification can be achieved without requiring entanglement or a complete characterization of the intervening quantum states and measurements. Based on the energy-bounded framework introduced in [Van Himbeeck et al., Quantum 1, 33 (2017)], our prepare-and-measure protocol is able to amplify the randomness of any public Santha-Vazirani source, requiring the smallest number of inputs and outcomes possible and being secure against quantum adversaries.
\end{abstract}

\maketitle

\section{Introduction}
A randomness amplification protocol (RAP) takes bit strings from a \emph{single} min-entropy source, potentially correlated to an adversary, and produces fully private and random bit strings (i.e. uniform and indpendent of the adeversary's information) at its output. Santha-Vazirani (SV) sources \cite{santha1986generating}, a largely studied class of min-entropy sources, model processes in which bits $x_1,\dots,x_n$ are generated sequentially and where each bit $x_i$ can be correlated with all the preceding bits although not be completely determined by them, that is
\begin{align}\label{eq:sv-source}
\frac{1}{2}-\delta\leq p(x_i \mid x_1,\dots,x_{i-1})\leq \frac{1}{2}+\delta
\end{align}
for some $0\leq\delta<1/2$.
It is a well-known result in classical information theory that there is no \emph{deterministic} RAP for the class of SV sources \cite{santha1986generating}.

In a long line of research starting with Colbeck's and Renner's seminal work \cite{colbeck2012free}, quantum Bell nonlocality has been harnessed for the design of RAPs \cite{colbeck2012free,gallego2013full,brandao2016realistic,kessler2020device,foreman2020practical}. In these protocols, the input min-entropy source is used to choose the settings in a Bell test. If a Bell violation is observed, this \emph{certifies} that an arbitrarily random string can be extracted from the measurement outcomes; otherwise, the protocol is aborted.  This certification is \emph{device-independent} (DI), meaning that no assumption about the shared quantum state and local quantum measurements is made.

It is an elementary fact that entanglement and, hence, nonlocality are not needed to have randomness amplification in quantum theory. Namely, for a RAP that simply disregards the bits from the input min-entropy source and measures $\sigma_Z$ eigenstates in the eigenbasis of $\sigma_X$, the theory predicts a sequence of outcomes which is uniformly distributed and independent of anything else. In real implementations, however, states are never pure and measurements are never projective. This opens up the possibility of attacks by adversaries with some control of the additional degrees of freedom \cite{gerhardt2011full}. Therefore, a high level of trust in the characterization of the intervening quantum devices has to necessarily go into the security claims for such a \emph{device-dependent} RAP. On the other hand, it is also well-known that in the fully device-independent scenario, no randomness can be certified in prepare-and-measure setups, without entanglement. Naturally, this leads to the question: 
\begin{quote}
which, ideally minimal, set of assumptions allows for randomness amplification in a prepare-and-measure scenario?
\end{quote}
 The \emph{semi-device-independent} (semi-DI) paradigm allows to study precisely this kind of question. Semi-DI protocols \cite{pawlowski2011semi,zhou2015semi,brask2017megahertz,van2017semi,miklin2020semi,tavakoli2020semi} introduce some assumption about the intervening quantum states or measurements in order to lower the implementation's requirements while, at the same time, not having the pitfalls of a complete device-dependent protocol. 
 
In the semi-DI prepare-and-measure scheme introduced in \cite{van2017semi}, a bound on the energy or, more generally,
on the expectation value that a physically motivated observable takes on the otherwise uncharacterized prepared states is assumed. For example, in quantum optics setups such an observable could be the number of photons, the energy in some subset of the frequency modes, etc. In this work, we prove that randomness amplification can be certified in this semi-DI setting.

Our prepare-and-measure RAP, which results from porting the DI RAP of \cite{kessler2020device} to the semi-DI setting of \cite{van2017semi}, works for any public SV source \footnote{In this context, we say that a source $\SV$ is \emph{public} if, \textbf{after} manufacturing the protocol's devices, the adversary can have access to the bits produced by $\SV$ (i.e., the inputs to the device).} with bias $\delta<1/2$, producing bits with arbitrarily small bias, $\delta\to 0$, with the minimum number of inputs and outcomes possible and being secure against quantum adversaries. 

\paragraph{Related work.} The use of SV sources in a semi-DI setting was first considered in \cite{zhou2015semi}. The authors provide a randomness generation protocol based on a $2\to 1$ quantum random access code (QRAC) in which a qubit bound on the dimension of the prepared states is assumed. It is shown that fresh randomness can be generated if the inputs to the protocol are chosen with an SV source with bias $\delta < 0.1358$. In a subsequent work \cite{PhysRevA.94.032318}, the authors proved that a higher randomness generation rate can be achieved with a $3\to 1$ QRAC but at the expense of reducing the tolerated input bias to $\delta<0.103$. 

\medskip

This paper is organized as follows. First, we review the semi-DI framework based on energy bounds introduced in \cite{van2017semi}. Next, we describe the randomness amplification scenario using the said framework, and the assumptions that we make for the security proof. After that, we state the main technical contribution of this work: that nonzero conditional min-entropy can be certified in the energy-bounded semi-DI setting even if the input choices are taken from an arbitrarily bias SV source. Finally, we present our main result: a semi-DI protocol achieving full randomness amplification.


\section{The semi-DI framework of \cite{van2017semi}}\label{sec:semi-di-framework}
In the semi-DI framework introduced in \cite{van2017semi}, the basic setup comprises a preparation box $\prepBox$ with binary inputs and a measurement box $\msrmntBox$ with binary outputs. On input $x\in\{0,1\}$, $\prepBox$ prepares a quantum state $\rho^x$ and sends it to $\msrmntBox$ which performs some binary measurement $\{M_a\}_a$ on it, producing the output $a\in\{0,1\}$. The object of interest is the \emph{behaviour} $\prob_{A|X} = \{\prob_{A|x}\}_x$, i.e. the family of probability distributions of the output $A$ conditioned on the input $x$, for all $x\in\{0,1\}$. Clearly, if no assumption is made about $\prepBox$ or $\msrmntBox$, any behaviour $\prob_{A|X}$ can arise in such a setting. For example, $\prepBox$ can just send the input $x$ encoded in one of two orthogonal states and let $M$ locally sample the target $\prob_{A|x}$ after perfectly distinguishing between them.
Unlike the dimension bound of the first semi-DI protocols, in this framework the set of behaviours is restricted by introducing the assumption of a bound 
\begin{align}\label{eq:energy-bound}
\Tr{[H\rho^x]}\leq\omega_x
\end{align}
on the expectation value of some chosen observable $H$. The choice of observable may be motivated by a physical assumption on the setup, such as an energy bound on the prepared states. It is important to point out that there are no restrictions on the Hermitian operator $H$ other than having a nondegenerate smallest eigenvalue and a finite gap. These assumptions are quite natural in photonic setups, the ground (first excited) state being the vacuum (one-photon) state. W.l.o.g. we assume $0$ to be $H$'s smallest eigenvalue and $1$ the second to smallest. Intuitively, if both $\omega_0$ and $\omega_1$ are close to $0$, then both the prepared states $\rho^0$ and $\rho^1$ are close to $H$'s unique ground state and are, hence, hard to distinguish. For ease of reading, we will henceforth refer to $H$ as ``the energy''. Given an energy bound $\energyBound := (\omega_0,\omega_1)$, the thus allowed set of \emph{$\energyBound$-bounded behaviours} is 
\begin{align*}
\quantumSet{\energyBound}&:=\{\prob_{A|X}|\exists \{\rho^x\}_x,\{M_a\}_a \textrm{ s.t. } p(a|x)=\Tr[M_a\rho^x] \\
&\qquad\qquad\qquad\qquad   \textrm{ and } \Tr[H\rho^x]\leq \omega_x\}.
\end{align*}
%

Analogously to DI protocols, the successful execution of a semi-DI protocol in this framework is certified by the observation of a nonclassical behaviour. Much like in the Bell nonlocality setting, the classical behaviours $\classicalSet{\energyBound}\subseteq \quantumSet{\energyBound}$ in this semi-DI scenario are defined to be those which can be reproduced as a convex combination of deterministic behaviours, i.e. 
\begin{align}\label{eq:classical-set}
\classicalSet{\energyBound}&:=\{\sum_\lambda p(\lambda)\prob_{A|X,\lambda} \mid \exists~\prob_\Lambda,\{(\energyBound^\lambda,\prob_{A|X,\lambda})\}_{\lambda\in\Lambda} \textrm{ s.t. }  \nonumber\\
&\quad \prob_{A|X,\lambda} \in \{0,1\}^4\cap \mathcal{Q}_{\energyBound^\lambda} \textrm{ and }  \sum_\lambda p(\lambda)\omega^\lambda_{x}\leq\omega_x\}.
\end{align}
\section{Setting and assumptions for our randomness amplification protocol}\label{sec:setting-and-assumptions}
For our randomness amplification protocol, we consider a setting where we will use $n$ times in succession a public $\delta$-SV source $\SV_\delta$ (see Eq. \eqref{eq:sv-source}) and an untrusted device $D_\energyBound$, possibly manufactured by an adversary, Eve, made of two components: a preparation box $\prepBox_\energyBound$ and a measurement box $\msrmntBox$. 
Both the device and the source can depend on some classical side information $\lambda$ that the adversary holds. In particular, $\lambda$ can include all the bits (i.e. the history) produced by $\SV_\delta$ before Eve prepares $D_\energyBound$. Moreover, Eve can hold a quantum memory $E$ entangled with the states prepared by box $\prepBox_\energyBound$. During the execution of the protocol, $\SV_\delta$ produces the inputs $\mathbf{X}=X_1\dots X_n$ for the device which then, upon receiving the inputs,  produces the outputs $\mathbf{A}=A_1\dots A_n$. After the device has produced its outputs, the source produces another binary string $\mathbf{Z}=Z_1\dots Z_d$. See Fig. \ref{fig:setup} for a pictorial depiction of the setting.
\begin{figure}
     \centering
     \includegraphics[trim=1.5cm 5.5cm 1.5cm 1cm, clip,scale=0.3]{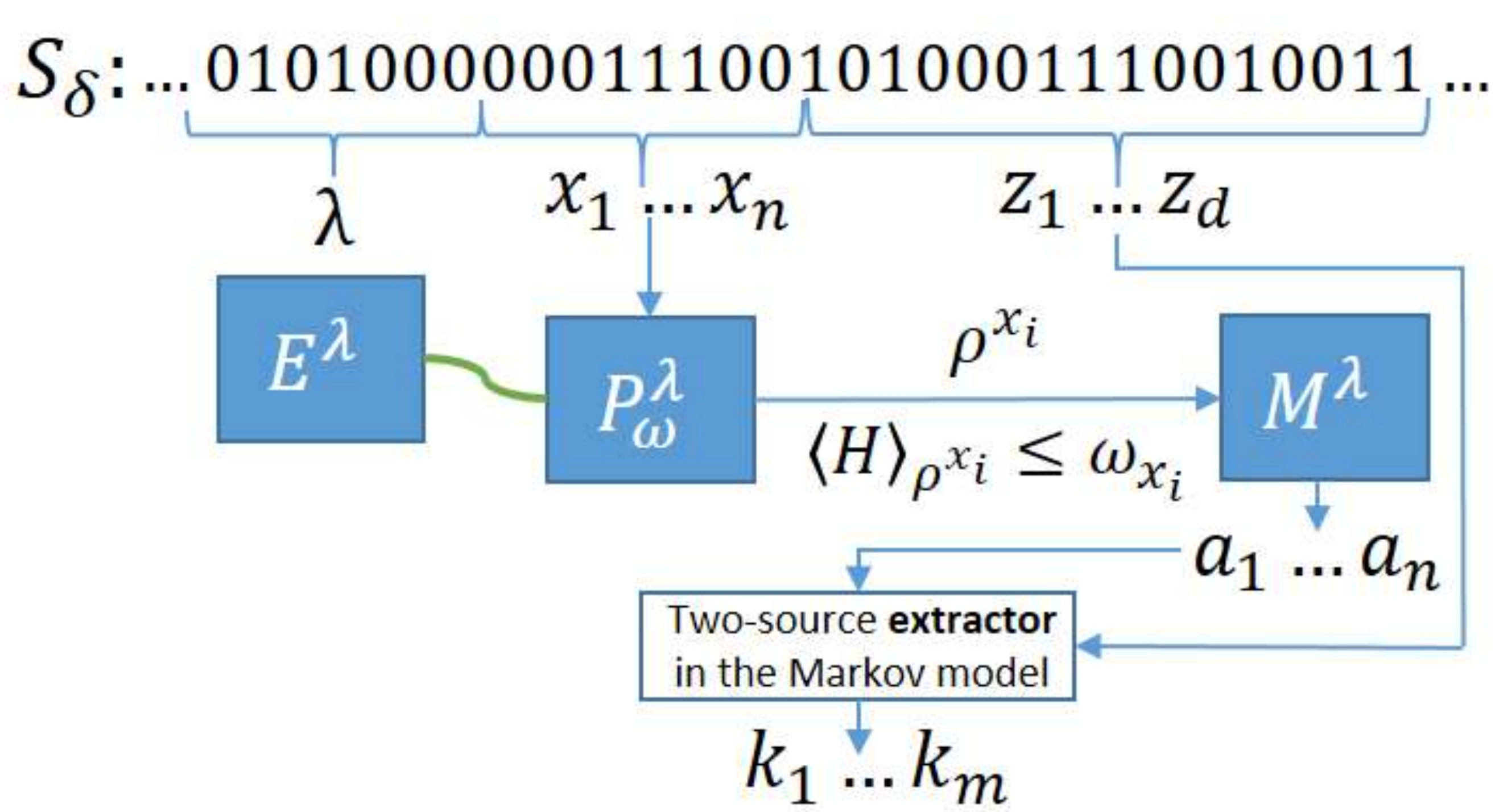}
     \caption{Pictorial depiction of the considered setting. Eve has classical side information $\lambda$ about the public SV-source $\SV_\delta$. She manufactures a prepare-and-measure device $D^\lambda_\energyBound=(P^\lambda_\energyBound,M^\lambda)$ and keeps a quantum memory $E^\lambda$ entangled with the $\energyBound$-bounded states prepared by $P_\energyBound$. A string of bits $\mathbf{x}\in\{0,1\}^n$ from $\SV_\delta$ is fed to the device, which produces $\mathbf{a}\in\{0,1\}^n$. Finally, $\mathbf{a}$ together with another string from $\SV_\delta$, $\mathbf{z}\in\{0,1\}^d$, are fed to a two-source extractor in the Markov model \cite{arnon2016quantum}, which outputs the final string $\mathbf{k}\in\{0,1\}^m$.\label{fig:setup}}
\end{figure}

We assume that:

\begin{enumerate}[label=(\subscript{A}{\arabic*})]
\item\label{assumption:avg-energy} There exists a Hermitian operator $H$ with lowest nondegenerate eigenvalue $0$ and unit gap such that for all $i\leq n$, $\mathbf{x}\in\{0,1\}^i$ and $\mathbf{a}\in\{0,1\}^{i-1}$, $$\Tr[H\rho^{x_i|x_1\dots x_{i-1},\mathbf{a},\lambda}]\leq \omega_{x_i},$$ with $\rho^{x_i|x_1\dots x_{i-1},\mathbf{a}}$ the state prepared by $P_\energyBound$ on round $i$ when $X_1\dots X_i=\mathbf{x}$ and $A_1\dots A_{i-1}=\mathbf{a}$.
\item\label{assumption:no-entanglement} There is no entanglement between $P_\energyBound$ and $M$.
\item The adversary only has classical side information, $\lambda$, about the SV source $\SV_\delta$.
\item\label{assumption:markov-chain} All dependence between $D_\energyBound$ and the source $\SV_\delta$ is contained in the adversary's side information. More formally, 
we assume that while the device produces outputs, it holds that $$I(A_1\dots A_{i-1}:X_i|X_1\dots X_{i-1}E,\lambda) = 0$$ and, after the device is done, it holds that $$I(\mathbf{Z}:\mathbf{A}|\mathbf{X}E,\lambda) = 0$$
with $I(\cdot : \cdot|\cdot)$ the conditional mutual information.
\end{enumerate}

Assumptions \ref{assumption:avg-energy} and \ref{assumption:no-entanglement} come from the framework in \cite{van2017semi}, where the former is referred to as \emph{max-average assumption}. Assumption \ref{assumption:markov-chain}, coming from the DI RAP in \cite{kessler2020device}, can be seen as the restriction that Eve does not have access to $D_\energyBound$ or $\SV_\delta$ once the protocol commences.

\section{Single-round min-entropy from an MDL-like inequality violation}\label{sec:single-round-min-entropy}
In this section, we show that for every bias $\delta$ of the input SV source and for every value $\lambda$ of the adversary's side information such that,
\begin{align*}
\frac{1}{2}-\delta \leq p(x|\lambda) \leq \frac{1}{2}+\delta
\end{align*}
there exist $\energyBound$-bounded behaviours $\prob_{A|X}$ for which the conditional min-entropy $H_\mathrm{min}(A|X,E,\lambda)$ is nonzero.

First of all, notice that implicit in the definition of the classical behaviours in Eq. \eqref{eq:classical-set} there is the assumption that the choice of preparation $x$ is independent of the shared randomness $\lambda$. In our randomness amplification scenario, however, the devices prepared by Eve can be correlated with the inputs given by the SV source and, hence, observation of a behaviour outside $\classicalSet{\energyBound}$ may not necessarily imply the nonexistence of a classical explanation.
For example, in the Bell nonlocality setting, if one allows the inputs $(x,y)\in\{0,1\}^2$ in a CHSH test to be correlated with the devices via some shared random variable $\Lambda$ such that $0.1465\lesssim p(x,y|\lambda)\lesssim 0.2845$, an observation of the Tsirelson bound can be classically explained \cite{sheridan2013bell}. The theory of \emph{measurement-dependent local} (MDL) distributions, introduced in \cite{putz2014arbitrarily,putz2016measurement}, was developed to study precisely this kind of classical explanations. Although originally conceived for Bell nonlocality scenarios, the analogous of MDL distributions can straightforwardly be defined in our prepare-and-measure setting. Concretely, the thus prescribed set of MDL-like distributions is 
\begin{align*}
\classicalSet{\energyBound}^\delta&:=\{\sum_\lambda p(\lambda)\prob_{AX|\lambda} \mid \frac{1}{2}-\delta\leq p(x|\lambda)\leq \frac{1}{2}+\delta \text{ and }\\
&\qquad\qquad\qquad\prob_{A|X,\lambda}\in\classicalSet{\energyBound}\}.
\end{align*}
As it follows from the results in \cite{putz2014arbitrarily,putz2016measurement}, 
$\classicalSet{\energyBound}^\delta$ is a polytope and, as in standard Bell scenarios, one can hence certify nonclassicality via the observation of an MDL-like inequality violation. 

The main technical contribution of this work, whose derivation we defer to Appendix \hyperref[appendix:mdl-ineq]{A}, is the following familiy of MDL-like inequalities, index by $\energyBound$ and $\delta$, for the energy-bounded prepare-and-measure scenario of \cite{van2017semi}:
\begin{lemma}[MDL-like inequality.]\label{lemma:mdl-ineq} For all $\prob_{AX} \in \classicalSet{\energyBound}^\delta$, it holds that
\begin{align}\label{eq:mdl-inequality}
\mdlIneq{\energyBound}(\prob_{AX}) := \mu_{\energyBound,\delta}\cdot p(a=x)-\frac{p(a\neq x)}{\mu_{\energyBound,\delta}}\geq B_{\energyBound,\delta},
\end{align}
with $B_{\energyBound,\delta}:=(1/2-\delta)(\mu_{\energyBound,\delta}+\frac{1}{\mu_{\energyBound,\delta}})(1-\omega_0-\omega_1)-\frac{1}{\mu_{\energyBound,\delta}}$ and $\mu_{\energyBound,\delta}:=(1/4-\delta^2)\omega_0\omega_1$.
\end{lemma}

Lemma \ref{lemma:mdl-ineq} says that distributions $\prob_{AX}$ for which $\mdlIneq{\energyBound}(\prob_{AX})<B_{\energyBound,\delta}$ cannot be reproduced deterministic $\energyBound$-bounded behaviours correlated with the $\delta$-SV source. This implies that, in the semi-DI framework of \cite{van2017semi}, if the inputs to the preparation box $\prepBox_\energyBound$ are taken from a $\delta$-SV source, the observation of a distribution $\prob_{AX}$ violating Eq. \eqref{eq:mdl-inequality} certifies that there is some degree of intrinsic randomness in the measurement box $\msrmntBox$'s outcomes. 

Using standard techniques borrowed from the DI setting \cite{bancal2014more} together with the SDP characterization of the sets $\quantumSet{\energyBound}$ given in \cite[Thm. 1]{van2019correlations}, in Appendix \hyperref[appendix:min-entropy]{B} we derive SDP lower bounds 
\begin{align}\label{eq:min-entropy-lower-bounds}
H_\mathrm{min}(A|X,E,\lambda)\geq \eta (\mdlIneq{\energyBound}^*)
\end{align}
to the single-round conditional min-entropy as a function of the violation $\mdlIneq{\energyBound}^*$ of Eq. \ref{eq:mdl-inequality}. In Fig. \ref{fig:min-entropy-comb}, we plot the maximum conditional min-entropy that we can certify for a given bias $\delta\in[0,1/2)$ of the SV source, optimizing over $\energyBound\in[0,1]^2$ and violations of the corresponding $\mdlIneq{\energyBound}$. 
\begin{figure}[H]
     \centering
     \includegraphics[width=\columnwidth]{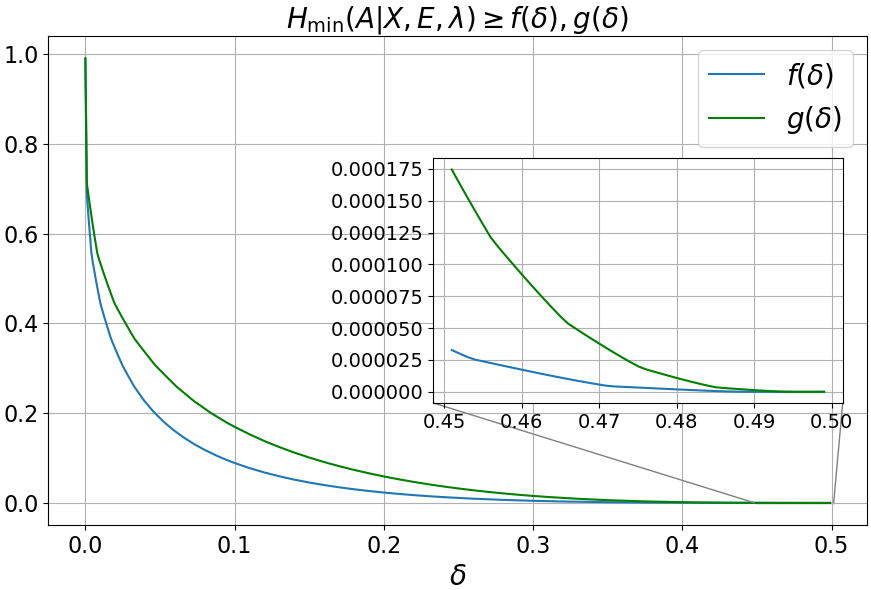}
\caption{Lower bounds to the single-round min-entropy certified by a violation of Eq. \ref{eq:mdl-inequality} as a function of the SV source's bias $\delta$ and for a suitable choice $\omega(\delta)$ of the ``energy'' bound in Eq. \ref{eq:energy-bound}. The {\color{blue}blue} curve (see Eq. \eqref{eq:min-entropy-lower-bound-general} in Appendix \hyperref[appendix:min-entropy]{B})  corresponds to the general bound. The {\color{darkgreen} green} lower bound (see Eq. \eqref{eq:min-entropy-lower-bound-uniform} in Appendix \hyperref[appendix:min-entropy]{B}) applies when, as it's customary (see e.g. \cite{colbeck2012free,zhou2015semi}), the observed distribution of the inputs is uniform, i.e. $\sum_\lambda p(x,\lambda)=1/2$.\label{fig:min-entropy-comb}}
\end{figure}
As expected, the amount of certifiable min-entropy decreases as the allowed correlation between the adversary and the source increases. Nevertheless, albeit very small, nonzero min-entropy can be certified even for $\delta\to 1/2$. On the other hand, if the adversary is uncorrelated with the source ($\delta=0$), one bit of min-entropy is reached.

\section{Randomness amplification protocol}\label{sec:protocol}
Our semi-DI RAP, given in Protocol \ref{protocol:rap}, is an adaptation of the DI RAP in \cite[Protocol 2]{kessler2020device} to our semi-DI scenario. It consists of two parts. In the first part, entropy is accumulated by performing a series of $n$ MDL-like experiments. In these rounds, we draw inputs $X_i$ from  the public $\delta$-SV  source and feed them to the device $D_\energyBound$ which produces outputs $A_i$. Let 
\begin{align*}
\mathrm{freq}_\mathbf{AX}(a,x):=\frac{|\{i\leq n|(A_i,X_i)=(a,x)\}|}{n}
\end{align*}
be the empirical frequencies after $n$ entropy accumulation rounds. We decide whether to abort or not by comparing $\mdlIneq{\energyBound}(\mathrm{freq}_\mathbf{AX})$ with $(I_{\rm exp}+\gamma_{\mathrm{est}})$, where $I_{\rm exp}$ is the expected violation of the MDL ineq. in Eq. \eqref{eq:mdl-inequality} and $\gamma_{\mathrm{est}}\in \left ( 0,I_\mathrm{exp} \right )$ is the tolerated deviation from such value. In  the  second  part,  we  draw  another string  from  the  $\delta$-SV  source  and  use  this  string,  as  well  as the  output  from  the  entropy  accumulation  part,  as  inputs  for a quantum-proof randomness extractor in the Markov model strong in the second input \cite{arnon2016quantum}. The extractor then produces the final output $\mathbf{K}$. 

\begin{algorithm}[H]
	\floatname{algorithm}{Protocol}
	\raggedright
	\caption{Randomness Amplification Protocol}
	\label{protocol:rap}
	\begin{algorithmic}[1]
		\STATEx \textbf{Arguments:} 
		\STATEx\hspace{\algorithmicindent} $\SV_\delta$ -- $\delta$-SV source.
		\STATEx\hspace{\algorithmicindent} $D_\energyBound$ -- untrusted device $D_\energyBound$ made of two components: a preparation box $\prepBox_\energyBound$ and a measurement box $\msrmntBox$.
		\STATEx\hspace{\algorithmicindent} $n \in \mathbb{N}_+$ -- number of rounds.
		\STATEx\hspace{\algorithmicindent} $I_{\mathrm{exp}}$ -- expected violation of Eq. \eqref{eq:mdl-inequality}.
		\STATEx\hspace{\algorithmicindent} $\gamma_{\mathrm{est}} \in \left ( 0,I_\mathrm{exp} \right )$ -- width of statistical confidence interval for the estimation test.
		\STATEx\hspace{\algorithmicindent} $\mathrm{Ext}:\{0,1\}^{n}\times\{0,1\}^d\rightarrow\{0,1\}^m$ -- $(k_{1}, k_{2}, \varepsilon_{\mathrm{ext}})$ quantum-proof randomness extractor in the Markov model which is strong in the second input.
	
		\STATEx
	
		\STATEx \textbf{Entropy Accumulation:}
		\STATE For every round $i \in \{1,\dots,n\}$ do:
		\STATE\hspace{\algorithmicindent} Draw a bit $X_i$ from $\SV_\delta$.
		\STATE\hspace{\algorithmicindent} Feed $X_i$ to $\prepBox_\energyBound$ and record $\msrmntBox$'s output $A_i$.
		\STATE Abort the protocol if $\mdlIneq{\energyBound}(\mathrm{freq}_\mathbf{AX}) > (I_\mathrm{exp} + \gamma_{\mathrm{est}})$.
		\STATEx
		\STATEx \textbf{Randomness Extraction:}
		\STATE Draw a bit string $\mathbf{Z}$ of length $d$ from $\SV_\delta$.
		\STATE Use $\mathrm{Ext}$ to create $\mathbf{K} = \mathrm{Ext}(\mathbf{A}, \mathbf{Z})$.
	\end{algorithmic}
\end{algorithm}

Our main result is Theorem \ref{thm:soundness-and-completeness}, which states that for every $\delta$-SV source, there is a choice of parameters for Protocol \ref{protocol:rap} such that with arbitrarily high probability it does not abort and produces a bit string $\mathbf{K}$ which is $\epsilon$-close (in trace distance) to being uniformly distributed and independent of all the adversary's side information.

\begin{theorem}\label{thm:soundness-and-completeness}
Given any public $\delta$-SV source $\SV_\delta$, with $0\leq\delta<1/2$, there exists an energy bound $\energyBound\in [0,1]^2$ and an achievable violation $\mdlIneq{\energyBound}^*$ of Eq. \eqref{eq:mdl-inequality} such that for every $I_\mathrm{exp}\leq \mdlIneq{\energyBound}^*$:
\begin{itemize}[leftmargin=*]
\item There exists a device $D_\energyBound$ satisfying assumptions \ref{assumption:avg-energy}-\ref{assumption:no-entanglement} such that Protocol \ref{protocol:rap} does not abort with probability
\begin{align}
1-\Pr[\mathrm{Abort}]\geq 1-2^{-O(n\gamma_{\mathrm{est}}^2)}.\tag{completeness}
\end{align}
\item For any desired security parameter $\epsilon_s\in(0,1)$ and any desired length $m$ of the output string $\mathbf{K}$, there exists a number $n$ of rounds and a number $d$ of additional bits from $\SV_\delta$ such that if a device $D_\energyBound$ and the source $\SV_\delta$ satisfy assumptions \ref{assumption:avg-energy}-\ref{assumption:markov-chain}, then
\begin{align}
\frac{1}{2}(1-\Pr[\mathrm{Abort}])||\rho_{\mathbf{K}\Sigma}-\rho_{U_m}\otimes\rho_{\Sigma}||_\mathrm{tr}\leq \epsilon_s,\tag{soundness}
\end{align} 
where $\Sigma=E\mathbf{X}\mathbf{Z}$ is Eve's side information and $\rho_{U_m}$ is the maximally mixed state of $m$ qubits.
\end{itemize}
%
%
\end{theorem}

\begin{proof}[Proof sketch] The proof, which we defer to Appendix \hyperref[appendix:proof-main-result]{C}, goes along the same lines as that for the DI RAP of \cite[Protocol 2]{kessler2020device}. 
Completeness follows from the existence of $\energyBound$-behaviours violating Eq. \eqref{eq:mdl-inequality} for every bias $\delta$ of the input's distribution and from applying a Hoeffding bound to sufficiently many independent copies of such distributions. As for the soundness, it follows from using the Entropy Accumulation Theorem (EAT) \cite{dupuis2016entropy} to go from the lower bounds to $H_\textrm{min}(A_i|X_i,E)$ in Eq. \eqref{eq:min-entropy-lower-bounds} to lower bounds to the $n$-round conditional smooth min-entropy $H_\textrm{min}^{\epsilon_s}(\mathbf{A}|\mathbf{X},E)$ (we follow the techniques in \cite{8935370}) and then proving the existence of suitable arguments for the extractor to produce, from $\mathbf{A}$ and $\mathbf{Z}$, the desired number $m$ of $\epsilon_s$-secure bits. 
\end{proof}

\section{Discussion}\label{sec:discussion}
In this work, we have proven that randomness amplficiation can be achieved in a prepare-and-measure scenario with the assumption of an energy-bound on the otherwise uncharecterized prepared states. In addition to being the first semi-DI RAP, by tolerating the whole range of biases $\delta\in[0,1/2)$ our result significantly improves over previous works which considered the use of $\delta$-SV sources in a semi-DI setting \cite{zhou2015semi,PhysRevA.94.032318}.
We expect our techniques, chiefly those leading to Eq. \ref{eq:mdl-inequality}, to be useful in the design of semi-DI protocols incorporating SV sources in other semi-DI schemes, such as the standard dimension-bounded or the recently introduced based on ``restricted-distrust'' \cite{tavakoli2021semi}.

\bigskip

\emph{Acknowledgments.} We acknowledge financial support from the ERC AdG CERQUTE, the EU project QRANGE, the AXA Chair in Quantum Information Science, the Government of Spain (FIS2020-TRANQI and Severo Ochoa CEX2019-000910-S), Fundaci\'o Cellex, Fundaci\'o Mir-Puig and Generalitat de Catalunya (CERCA, AGAUR SGR 1381).

\bibliographystyle{unsrtnat}

%
\pagebreak

\begin{widetext}

\section{Appendix A: Proof of Lemma \ref{lemma:mdl-ineq}}
\label{appendix:mdl-ineq}
Let us first recall the definition of the set of MDL-like classical behaviours for our semi-DI scenario with $\delta$-SV sources:
\begin{align*}
\classicalSet{\energyBound}^\delta&:=\{\sum_\lambda p(\lambda)\prob_{AX|\lambda} \mid \frac{1}{2}-\delta\leq p(x|\lambda)\leq \frac{1}{2}+\delta \text{ and }\prob_{A|X,\lambda}\in\classicalSet{\energyBound}\}.
\end{align*}
Behaviours outside this set cannot be reproduced by convex combinations of $\energyBound$-bounded deterministic strategies even if they are allowed to be correlated with the choice of inputs. As it follows from the results in \cite{putz2014arbitrarily,putz2016measurement}, 
$\classicalSet{\energyBound}^\delta$ is a polytope whose set of vertices $\mathcal{V}(\classicalSet{\energyBound}^\delta)$ satisfies:

\begin{align}\label{eq:vertices-mdl-set}
\mathcal{V}(\classicalSet{\energyBound}^\delta)\subseteq\{\prob_{AX}|p(a,x)=p(a|x)v(x) \land \prob_{A|X}\in\mathcal{V}(\classicalSet{\energyBound}) \land \mathbf{v}_X\in\{(1/2+\delta,1/2-\delta),(1/2-\delta,1/2+\delta)\}\}
\end{align}
with
\begin{align}\label{eq:vertices-local-set}
&\mathcal{V}(\classicalSet{\energyBound})=\{(1,0,1,0),(0,1,0,1),(1,0,1-(\omega_0+\omega_1),\omega_0+\omega_1),&\nonumber\\
&\qquad\qquad (0,1,\omega_0+\omega_1,1-(\omega_0+\omega_1)),(1-(\omega_0+\omega_1),\omega_0+\omega_1,1,0),(\omega_0+\omega_1,1-(\omega_0+\omega_1),0,1)\}&
\end{align}
the set of vertices of $\classicalSet{\energyBound}$ \cite{van2017semi}. We now restate and prove Lemma \ref{lemma:mdl-ineq}.

\medskip
\textbf{Lemma 1} For all $\prob_{AX} \in \classicalSet{\energyBound}^\delta$ it holds that
\begin{align}\label{eq:mdl-inequality-app}
\mdlIneq{\energyBound}(\prob_{AX}) := \mu_{\energyBound,\delta}[p(a==x)]-\frac{1}{\mu_{\energyBound,\delta}}[p(a\neq x)]\geq B_{\energyBound,\delta},
\end{align}
with 
\begin{align}
B_{\energyBound,\delta}:=[\mu_{\energyBound,\delta}+\frac{1}{\mu_{\energyBound,\delta}}](1/2-\delta)(1-\omega_0-\omega_1)-\frac{1}{\mu_{\energyBound,\delta}}
\end{align}
and $\mu_{\energyBound,\delta}:=(1/4-\delta^2)\omega_0\omega_1$.

\begin{proof}
Let $\prob_{AX} \in \classicalSet{\energyBound}^\delta$. Then,
\begin{align*}\label{eq:proof-mdl-bound}
\mdlIneq{\energyBound}(\prob_{AX}) &= \sum_\lambda p(\lambda)\sum_x p(x|\lambda)[\mu_{\energyBound,\delta}p(a==x|x,\lambda)-\frac{1}{\mu_{\energyBound,\delta}}p(a==1-x|x,\lambda)]\\
&=\sum_\lambda p(\lambda)\sum_x p(x|\lambda)[\mu_{\energyBound,\delta}p(a==x|x,\lambda)-\frac{1}{\mu_{\energyBound,\delta}}(1-p(a==x|x,\lambda))]\\
&=\sum_\lambda p(\lambda)\sum_x p(x|\lambda) (\mu_{\energyBound,\delta}+\frac{1}{\mu_{\energyBound,\delta}})[p(a==x|\lambda)]-\frac{1}{\mu_{\energyBound,\delta}}\\
&\geq \sum_\lambda p(\lambda)(\frac{1}{2}-\lambda)(\mu_{\energyBound,\delta}+\frac{1}{\mu_{\energyBound,\delta}})[p(0|0,\lambda)+p(1|1,\lambda)]-\frac{1}{\mu_{\energyBound,\delta}}\\
&\geq (\frac{1}{2}-\delta)(\mu_{\energyBound,\delta}+\frac{1}{\mu_{\energyBound,\delta}})(1-\omega_0-\omega_1)-\frac{1}{\mu_{\energyBound,\delta}}=B_{\energyBound,\delta},
\end{align*}
where in the last line we have used that $p(0|0,\lambda)+p(1|1,\lambda)\geq 1-(\omega_0+\omega_1)$ for all the vertices of $\classicalSet{\energyBound}$ in Eq. \eqref{eq:vertices-local-set} (which can be seen by a simple inspection) and therefore, by linearity, it also holds for all $\prob_{A|X,\lambda}\in \classicalSet{\energyBound}$.
\end{proof}

\section{Appendix B: Single-round min-entropy}\label{appendix:min-entropy}
In this section we quantify how much randomness can be certified as a function of the amount of violation of $\mdlIneq{\energyBound}$.

We consider a scenario in which the boxes $\prepBox$ and $\msrmntBox$ have been prepared by an eavesdropper Eve who holds some classical side information $\lambda$ about the $\delta$-SV source. In full generality, Eve chooses a quantum realization $q^\lambda=\{\{\$^{x,\lambda}\}_x,\rho,\{M^\lambda_a\}_a\}$ for the boxes, where $\{\$^{x,\lambda}\}_x$ are local quantum channels that box $\prepBox$ will apply to (w.l.o.g) a fixed initial state $\rho$ on input $x$ and $\{M^\lambda_a\}_a$ is the binary measurement to be performed on $\rho^{x,\lambda}=\$^{x,\lambda}(\rho)$ by box $\msrmntBox$. In addition, we let Eve hold a purification $\ket{\psi}$ of $\rho=\Tr_E[\ket{\psi}\bra{\psi}]$. Eve's aim is to guess $\msrmntBox$'s outcome when the input to $\prepBox$ was $x$ by using the result $c$ she gets from a measurement $\{E^{\lambda}_c\}_c$ on her share of $\ket{\psi}$. We interpret Eve's measurement as a state preparation. That is, when Eve's measurement result is $c$, occurring with probability $$p(c|x,\lambda)=\Tr[(\mathbb{I}\otimes E^
{\lambda}_c)\{\$^{x,\lambda}\otimes\mathbb{I}\}(\ket{\psi}\bra{\psi})],$$ 
box $\prepBox$ effectively prepares the state $$\rho^{x,\lambda}_c=\Tr_E[\left(\mathbb{I}\otimes E^{\lambda}_c\{\$^{x,\lambda}\otimes\mathbb{I}\}(\ket{\psi}\bra{\psi})\mathbb{I}\otimes (E^{\lambda}_c)^\dag\right)]/p(c|x,\lambda)$$
on input $x$. Notice that, since the channels $\$^{x,\lambda}$ are local, $p(c|0,\lambda)=p(c|1,\lambda)=p(c|\lambda)$. 

In this adversarial setting, two different instantiations of the energy constraint in Eq. \eqref{eq:energy-bound} were considered in \cite{van2017semi}:
\begin{align}
\sum_c p(c|\lambda)\Tr[H\rho^{x,\lambda}_c]=\sum_c p(c|\lambda)\omega_x^{\lambda,c}=\Tr[H\rho^{x,\lambda}]&\leq \omega^{\rm avg}_x\tag{max-avg-energy}\\
\max_c \Tr[H\rho^{x,\lambda}_c]&\leq \omega^{\rm pk}_x\tag{max-peak-energy}
\end{align}
Clearly, the max-peak assumption is the strongest and it is not hard to see that if $(\omega^{\rm pk}_0,\omega^{\rm pk}_1)\geq (1,1)$ then the set of allowed correlations is completely unrestricted. In this work we will only assume a bound $\energyBound^{\rm avg}=(\omega^{\rm avg}_0,\omega^{\rm avg}_1)$ on the average energies.

%
%
%
%
%

Let us now assume that a value $\mdlIneq{\energyBound^{\rm avg}}^*<B_{\energyBound^{\rm avg},\delta}$ of Eq. \eqref{eq:mdl-inequality} is observed. Eve's semi-DI average guessing probability is given by
\begin{align}\label{eq:pguess}
& p_{{\rm guess}}(A|X,E,\lambda,\delta,\mdlIneq{\energyBound^{\rm avg}}^*):= \max_{\prob_{C|\lambda},\{\prob_{A|X,c,\lambda}\}_c,\{\energyBound^{\lambda,c}\}_{c}}\sum_{x}p(x|\lambda)\sum_{c}p(c|\lambda)\max_a p(a|x,c,\lambda)&&\nonumber\\
&\qquad\qquad~\textrm{subject to } &\nonumber\\
&\qquad\qquad\qquad\qquad \mdlIneq{\energyBound^{\rm avg}}(\prob_{AX|\lambda})\leq \mdlIneq{\energyBound^{\rm avg}}^* &\nonumber\\
&\qquad\qquad\qquad\qquad \prob_{A|X,\lambda,c}\in \quantumSet{\energyBound^{\lambda,c}} \text{ for all } c&\nonumber\\
&\qquad\qquad\qquad\qquad \sum_c p(c|\lambda)\omega^{\lambda,c}_c\leq \omega_x^{\rm avg}\text{ for all } x.
\end{align}

This optimization depends on $\prob_{X|\lambda}$ of which we only assume the the SV condition $$1/2-\delta\leq p(x|\lambda)\leq 1/2+\delta.$$ 
By noting that for all $\prob_{AX|\lambda}$ it holds that
\begin{align*}
\mdlIneq{\energyBound}(\prob_{AX|\lambda})&=\mu_{\energyBound^{\rm avg},\delta}[p(0,0|\lambda)+p(1,1|\lambda)]-\frac{1}{\mu_{\energyBound^{\rm avg},\delta}}[p(1,0|\lambda)+p(0,1|\lambda)]\\
&=\mu_{\energyBound^{\rm avg},\delta}[p(0|0,\lambda)p(0|\lambda)+p(1|1,\lambda)p(1|\lambda)]-\frac{1}{\mu_{\energyBound^{\rm avg},\delta}}[p(1|0,\lambda)p(0|\lambda)+p(0|1,\lambda)p(1|\lambda)]\\
&\geq 
(\frac{1}{2}-\delta)\mu_{\energyBound^{\rm avg},\delta}[p(0|0,\lambda)+p(1|1,\lambda)]-\frac{(\frac{1}{2}+\delta)}{\mu_{\energyBound^{\rm avg},\delta}}[p(1|0,\lambda)+p(0|1,\lambda)]\\
&=:I^\mathrm{lower}_{\energyBound^{\rm avg},\delta}(\prob_{A|X,\lambda})
\end{align*}
we get the following upper bound to Eq. \eqref{eq:pguess}
\begin{align*}
& p_{{\rm guess}}(A|X,E,\lambda,\delta,\mdlIneq{\energyBound^{\rm avg}}^*)\leq (\frac{1}{2}+\delta) \max_{\prob_{\Lambda'},\{\prob_{A|X,\lambda'}\}_{\lambda'},\{\energyBound^{\lambda'}\}_{\lambda'}}\sum_{x,\lambda'}p(\lambda')\max_a p(a|x,\lambda')&&\nonumber\\
&\qquad\qquad~\textrm{subject to } &\nonumber\\
&\qquad\qquad\qquad\qquad \sum_{\lambda'}p(\lambda')I^\mathrm{lower}_{\energyBound^{\rm avg},\delta}(\prob_{A|X,\lambda})\leq \mdlIneq{\energyBound^{\rm avg}}^* &\nonumber\\
&\qquad\qquad\qquad\qquad \prob_{A|X,\lambda'}\in \quantumSet{\energyBound^{\lambda'}} \text{ for all } \lambda'&\nonumber\\
&\qquad\qquad\qquad\qquad \sum_{\lambda'} p(\lambda')\omega^{\lambda'}_x\leq \omega_x^{\rm avg}\text{ for all } x.
\end{align*}
where we have introduced the random variable $\Lambda'$ distributed according to $p(\lambda')=\mathrm{Pr}[c=\lambda'|\lambda]$. 

From \cite[Prop. 1]{bancal2014more} it follows that for this maximization is enough to consider $|\{0,1\}^2|=4$ different values of $\Lambda'$ (one per each combination of inputs and outputs), and therefore
\begin{align}\label{eq:upper-bound-pguess}
& p_{{\rm guess}}(A|X,E,\lambda,\delta,\mdlIneq{\energyBound^{\rm avg}}^*)\leq \eta(\delta,\energyBound^{\rm avg},\mdlIneq{\energyBound^{\rm avg}}^*)
\end{align}
with
\begin{align}\label{eq:sdp-upper}
& \eta(\delta,\energyBound,I):=(\frac{1}{2}+\delta)\max_{\{(\tilde{\prob}_{A|X,\lambda'},\energyBound^{\lambda'})\}_{\lambda'\in\{0,1\}^2}}\sum_{a,x} \tilde{p}_{(a,x)}(a|x)&\nonumber\\
&\qquad\qquad\qquad\textrm{subject to } &\nonumber\\
&\qquad\qquad\qquad\qquad \sum_{\lambda'}I^\mathrm{lower}_{\energyBound^{\rm avg},\delta}(\prob_{A|X,\lambda'})\leq I &\nonumber\\
&\qquad\qquad\qquad\qquad  \tilde{\prob}_{A|X,\lambda'}\in \tilde{\quantumSet{\energyBound^{\lambda'}}} \text{ for all } \lambda'&\nonumber\\
&\qquad\qquad\qquad\qquad  \sum_{\lambda'}\energyBound^{\lambda'}_x \leq \omega_x\text{ for all } x &\nonumber\\
&\qquad\qquad\qquad\qquad  \sum_{a,\lambda'} \tilde{p_{\lambda'}}(a|x) = 1.
\end{align}
where we have absorbed the weights $p(\lambda')$ into the normalizations of $\tilde{\prob}_{A|X,\lambda'}$. Finally, the fact that Eq. \eqref{eq:sdp-upper} is an SDP follows from the SDP characterization of the sets $\quantumSet{\energyBound}$ given in \cite[Thm. 1]{van2019correlations}.

\medskip

In order to get a lower bound on the maximum amount of conditional min-entropy that can be certified in this scenario for given bias $\delta$ of the input SV source, we need to optimize over violations of the MDL inequality in Eq. \eqref{eq:mdl-inequality}. However, as before, the value of Eq. \eqref{eq:mdl-inequality} depends on $\prob_{X|\lambda}$ of which we only assume the the SV condition $1/2-\delta\leq p(x|\lambda)\leq 1/2+\delta$. By noting that for all $\prob_{AX|\lambda}$ it holds that
\begin{align*}
\mdlIneq{\energyBound}(\prob_{AX|\lambda})&=\mu_{\energyBound^{\rm avg},\delta}[p(0,0|\lambda)+p(1,1|\lambda)]-\frac{1}{\mu_{\energyBound^{\rm avg},\delta}}[p(1,0|\lambda)+p(0,1|\lambda)]\\
&=\mu_{\energyBound^{\rm avg},\delta}[p(0|0,\lambda)p(0|\lambda)+p(1|1,\lambda)p(1|\lambda)]-\frac{1}{\mu_{\energyBound^{\rm avg},\delta}}[p(1|0,\lambda)p(0|\lambda)+p(0|1,\lambda)p(1|\lambda)]\\
&\leq(\frac{1}{2}+\delta)\mu_{\energyBound,\delta}[p(0|0,\lambda)+p(1|1,\lambda)]-\frac{(\frac{1}{2}-\delta)}{\mu_{\energyBound,\delta}}[p(1|0,\lambda)+p(0|1,\lambda)]=: \mdlIneq{\energyBound}^{\rm upper}(\prob_{A|X,\lambda})
\end{align*}
we get the following lower bound on the amount of certifiable single-round min-entropy:
\begin{align}\label{eq:min-entropy-lower-bound-general}
\Hmin(A|X,E,\lambda,\delta) &:= -\log_2 p_{{\rm guess}}(A|X,E,\lambda,\delta) \nonumber\\
&\geq f(\delta)
\end{align}
with
\begin{align}
f(\delta)&:= -\log_2 \eta(\delta,\energyBound^{\rm avg}(\delta),\mdlIneq{\energyBound^{\rm avg}}^{\rm upper,opt})\\
\mdlIneq{\energyBound^{\rm avg}}^{\rm upper,opt}&:=\min_{\energyBound,\prob_{A|X}\in \quantumSet{\energyBound}} \mdlIneq{\energyBound}^{\rm upper}(\prob_{A|X})\\
\energyBound^{\rm avg}(\delta)&:=\argmin_{\energyBound}\min_{\prob_{A|X}\in \quantumSet{\energyBound}} \mdlIneq{\energyBound}^{\rm upper}(\prob_{A|X}).\label{eq:avg-omega}
\end{align}
A customary working assumption in both DI and semi-DI protocols dealing with SV sources is that the inputs, although arbitrarily correlated (up to $\delta$) with the devices, are seen as uniformly distributed by the honest users of the protocols (see e.g. \cite{colbeck2012free,zhou2015semi}). That is,
\begin{align*}
\sum_\lambda p(x,\lambda)=1/2.
\end{align*}
Intuitively, the motivation behind this assumption is that Eve does not want to reveal to the users her correlation with the source. Under this assumption, and by letting 
$$
\mdlIneq{\energyBound}^{\rm unif}(\prob_{A|X}):=\frac{\mu_{\energyBound,\delta}}{2}[p(0|0,\lambda)+p(1|1,\lambda)]-\frac{1}{2\mu_{\energyBound,\delta}}[p(1|0,\lambda)+p(0|1,\lambda)]
$$
we get the following (in general, better) min-entropy lower bound
\begin{align}\label{eq:min-entropy-lower-bound-uniform}
\Hmin(A|X,E,\lambda,\delta) &\geq g(\delta)
\end{align}
with
\begin{align*}
g(\delta)&:=-\log_2 \eta(\delta,\energyBound_{\rm unif}^{\rm avg}(\delta),\mdlIneq{\energyBound^{\rm avg}}^{\rm unif,opt})\\
\mdlIneq{\energyBound^{\rm avg}}^{\rm unif,opt}&:=\min_{\energyBound,\prob_{A|X}\in \quantumSet{\energyBound}} \mdlIneq{\energyBound}^{\rm unif}(\prob_{A|X})\\
\energyBound^{\rm avg}_{\rm unif}(\delta)&:=\argmin_{\energyBound}\min_{\prob_{A|X}\in \quantumSet{\energyBound}} \mdlIneq{\energyBound}^{\rm unif}(\prob_{A|X}).
\end{align*}

In Fig. \ref{fig:min-entropy} we plot $f(\delta)$ and $g(\delta)$. As can be seen in the plot, a nonzero min-entropy can be certified for $\delta\to 1/2$. Moreover, we approach one bit of min-entropy as the SV source's bias approaches $0$. 

\begin{figure}[H]
     \centering
     \subfloat[][]{\includegraphics[scale=0.45]{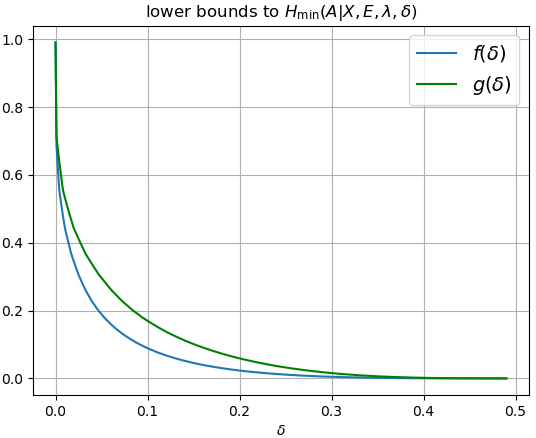}\label{fig:min-entropy-a-main}}
     \subfloat[][]{\includegraphics[scale=0.45]{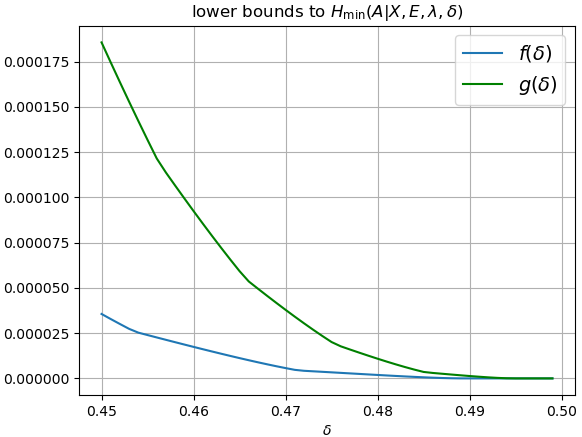}\label{fig:min-entropy-b-main}}
     \caption{Lower bounds to the single-round min-entropy certified by a violation of Eq. \ref{eq:mdl-inequality} as a function of the SV source's bias $\delta$ and for a suitable choice $\omega(\delta)$ of the ``energy'' bound in Eq. \ref{eq:energy-bound}. The {\color{blue}blue} curve corresponds to the general bound. The {\color{darkgreen} green} lower bound applies when, as it's customary (see e.g. \cite{colbeck2012free,zhou2015semi}), the observed distribution of the inputs is uniform, i.e. $\sum_\lambda p(x,\lambda)=1/2$. In (b) we zoom in the range $0.45\leq\delta< 0.5$, where the bounds are still nonzero (albeit very small).\label{fig:min-entropy}}
\end{figure}

We will prove the soundness and completeness of our RAP using the general bound in Eq. \eqref{eq:min-entropy-lower-bound-general}. The modifications required to use the better bound in Eq. \eqref{eq:min-entropy-lower-bound-uniform} are straightforward.

\section{Appendix C: Formal statement and proof of Theorem \ref{thm:soundness-and-completeness}}
Let us first restate Protocol \ref{protocol:rap} from the main text in an equivalent form which simplifies the analyisis with the Entropy Accumulation Theorem (EAT) \cite{dupuis2016entropy,arnon2019simple}.

\setcounter{algorithm}{0}
\begin{algorithm}[H]
	\floatname{algorithm}{Protocol}
	\raggedright
	\caption{Randomness Amplification Protocol}
	\label{protocol:rap2}
	\begin{algorithmic}[1]
		\STATEx \textbf{Arguments:} 
		\STATEx\hspace{\algorithmicindent} $\SV_\delta$ -- $\delta$-SV source
		\STATEx\hspace{\algorithmicindent} $D_\energyBound$ -- untrusted device $D_\energyBound$ made of two components: a preparation box $\prepBox_\energyBound$ and a measurement box $\msrmntBox$
		\STATEx\hspace{\algorithmicindent} $n \in \mathbb{N}_+$ -- number of rounds
		\STATEx\hspace{\algorithmicindent} $I_{\mathrm{exp}}$ -- upper bound on the expected violation of Eq. \eqref{eq:mdl-inequality-app}.
		\STATEx\hspace{\algorithmicindent} $\gamma_{\mathrm{est}} \in \left ( 0,I_\mathrm{exp} \right )$ -- width of statistical confidence interval for the estimation test 
		\STATEx\hspace{\algorithmicindent} $\mathrm{Ext}:\{0,1\}^{n}\times\{0,1\}^d\rightarrow\{0,1\}^m$ -- $(k_{1}, k_{2}, \varepsilon_{\mathrm{ext}})$ quantum-proof randomness extractor in the Markov model which is strong in the second input.
	
		\STATEx
	
		\STATEx \textbf{Entropy Accumulation:}
		\STATE For every round $i \in \{1,\dots,n\}$ do:
		\STATE\hspace{\algorithmicindent} Draw a bit $X_i$ from $\SV_\delta$.
		\STATE\hspace{\algorithmicindent} Feed $X_i$ to $\prepBox$ and record $\msrmntBox$'s output $A_i$.
		\STATE\hspace{\algorithmicindent} Set $C_i = w_{\delta,\energyBound}(A_i,X_i)$ for $w$ as defined in Eq. \eqref{eq:winning-function}. \label{step:last}
		\STATE Abort the protocol if $\bar{C} \equiv \frac{1}{n} \sum_j C_j > (I_\mathrm{exp} + \gamma_{\mathrm{est}})$. \label{step:abortion}
		\STATEx
		\STATEx \textbf{Randomness Extraction:}
		\STATE Draw a bit string $\mathbf{Z}$ of length $d$ from $\SV_\delta$.
		\STATE Use $\mathrm{Ext}$ to create $\mathbf{K} = \mathrm{Ext}(\mathbf{A}, \mathbf{Z})$.
	\end{algorithmic}
\end{algorithm}
The \emph{winning function} $w_{\delta,\energyBound}(\cdot)$ is defined as
\begin{align}\label{eq:winning-function}
w_{\delta,\energyBound}(A_i,X_i)=
\begin{cases}
\mu_{\energyBound,\delta} & (A_i,X_i)\in\{(0,0),(1,1)\}\\
-(\mu_{\energyBound,\delta})^{-1} & \text{otherwise.}
\end{cases}
\end{align}

The formal statement of our main result is:
\setcounter{theorem}{0}
\begin{theorem}\label{appendix:main-theorem}
Given any public $\delta$-SV source $\SV_\delta$, with $0\leq\delta<1/2$, and a device $D_{\energyBound}$ satisfying assumptions \ref{assumption:avg-energy}-\ref{assumption:markov-chain}, let $n$ be the number of rounds in Protocol~\ref{protocol:rap}, $\varepsilon_{\mathrm{s}}, \varepsilon_{\mathrm{EAT}} \in (0,1)$, $I_\mathrm{exp} \leq \mdlIneq{\energyBound^{\rm avg}}^{\rm upper,opt}$, $\gamma_{\mathrm{est}} \in \left ( 0,I_\mathrm{exp} \right )$ and $m,\varepsilon_{\mathrm{ext}} $ the parameters of the $(k_{1}, k_{2}, \varepsilon_{\mathrm{ext}})$-extractor used in Protocol~\ref{protocol:rap2}, with $k_1,k_2$ fulfilling Equation~\eqref{eq:ext_k_values}.
Then:
\begin{enumerate}
\item (Secrecy) Protocol~\ref{protocol:rap} produces a string $\mathbf{K}$ of length $m$ such that:
$$
\frac{1}{2}(1-\Pr[\mathrm{Abort}])||\rho_{\mathbf{K}\Sigma}-\rho_{U^m}\otimes\rho_{\Sigma}||_\mathrm{tr} \leq 6 \left(\varepsilon_{\mathrm{s}} +\varepsilon_{\mathrm{ext}} \right) + \varepsilon_{\mathrm{EAT}}\;,
$$
where $\Sigma=E\mathbf{X}\mathbf{Z}$ is Eve's side information and $\rho_{U^m}$ is the maximally mixed state of $m$ qubits.
\label{part:soundness}
\item (Completeness) There exists an honest implementation $D_{\energyBound^\mathrm{avg}(\delta)}$ of the device such that Protocol~\ref{protocol:rap} aborts with probability $\Pr[\mathrm{Abort}] \leq\exp \left( - \frac{2 n \mu_{\energyBound^\mathrm{avg}(\delta),\delta}^2\gamma_{\mathrm{est}}^2}{\left (1+\mu_{\energyBound^\mathrm{avg}(\delta),\delta}^2 \right )^2} \right)$ when using this device. \label{part:completeness}
\end{enumerate} 
\end{theorem}

In the following, we prove Theorem \ref{appendix:main-theorem}. The derivation will closely follow that of the main result of \cite{kessler2020device}, on which our result is based.

\subsection{Secrecy}
\subsubsection{EAT preliminaries}
Let us first recall the EAT's main concepts and statement.

\begin{definition}[EAT channels]\label{def:eat-channels}
A set of \emph{EAT channels} $\{\N_i\}_{i=1}^n$ is a collection of trace-preserving and completely-positive maps $\N_i:R_{i-1}\rightarrow A_iX_i\Ci R_i$ such that for every $i \in [n]$:
\begin{enumerate}
\item $A_i,X_i$ and $\Ci$ are finite dimensional classical
systems, $R_i$ is an arbitrary quantum system and $\Ci$ is the
output of a deterministic function of the classical registers
$A_i$ and $X_i$.\label{item:eat-ch-cond-1}
\item For any initial state $\rho_{R_0 E}$, the final state $\rho_{\mathbf{A}\mathbf{X}\mathbf{C}E} = \ptr{R_n}{((\N_n \circ \dots \circ \N_1)\otimes \id_{E}) \rho_{R_0 E}}$ fulfils the Markov chain condition $I(A_1\dots A_{i-1}:X_i|X_1\dots X_{i-1}E) = 0$ for every $i \in [n]$.\label{item:eat-ch-cond-2}
\end{enumerate}
\end{definition}

\begin{definition}[Min-tradeoff functions]\label{def:fmin}
Let $\{\N_i\}_{i=1}^n$ be a collection of EAT channels, 
$\C$ denote the common alphabet of the systems
$C_1,\dots,C_n$ and $\P_\C$ denote the set of probability distributions over $\mathcal{C}$. An affine function $f: \P_\C \rightarrow \R$ is a
\emph{min-tradeoff function} for the EAT channels
$\{\N_i\}_{i=1}^n$ if for each $i \in [n]$ it satisfies
\begin{equation*}\label{eq:fmin-definition}
f(\p) \leq \inf_{\sigma_{R_{i-1}R'}: \N_i(\sigma)_{\Ci}=\tau_{\p}} H(A_i|X_iR')_{\N_i(\sigma)},
\end{equation*} 
where $\tau_{\p}:=\sum_{c\in\mathcal{C}}p(c)\ket{c}\bra{c}$, $R'$ is a
register isomorphic to $R_{i-1}$ and the infimum over the empty set is
taken to be $+\infty$.
\end{definition}

The EAT, in a simplified version sufficient for our needs is:

\begin{theorem}[EAT~\cite{dupuis2016entropy}]
\label{thm:EAT}~\newline
Let $\{\N_i\}_{i=1}^n$ be a collection of EAT channels and let
$\rho_{\mathbf{A}\mathbf{X}\mathbf{C}E} = \ptr{R_n}{(( \N_n \circ \dots \circ
  \N_1)\otimes \id_{E})\rho_{R_0 E}}$ be the output state after the
sequential application of the channels $\{\N_i\otimes\id_{E}\}_i$ to
some input state $\rho_{R_0 E}$. Let $\Omega \subseteq \mathcal{C}^n$
be some event that occurs with probability $p_{\Omega}$ and let
$\rho_{|_\Omega}$ be the state conditioned on $\Omega$ occurring. Finally let $\epsmooth \in (0,1)$ and $f$ be a valid min-tradeoff function for $\{\N_i\}_i$. 
If for all $\CC\in\Omega$, with $\pr{\CC}>0$ there is some $t \in \R$
for which $f(\mathrm{freq}_{C^{n}})\geq t$, then
\begin{equation*}\label{eq:EAT-bound}
\Hmin^{\epsmooth}(\mathbf{A}|\mathbf{X} E)_{\rho_{|_{\Omega}}} > n t - n (\errV(f) + \errK(f)) - \errW,
\end{equation*}
where
\begin{equation*}
\mathrm{freq}_{C^{n}}(x) = \frac{|\{ i \in \{1, \dots, n\} : C_{i} = x \}|}{n},
\end{equation*}
\begin{equation*}
\errV(f) \leq \frac{\ln 2}{2\sqrt{n}} \left(\log\left(3\right) + \sqrt{ (\Max{f} - \Min{f})^2 + 2} \right)^2,
\end{equation*}

\begin{equation*}
\errK(f) \leq \frac{1}{6n(1-\sqrt{n})^3\ln 2}\,
2^{\sqrt{n}(1 + \Max{f} - \Min{f})}
\ln^3\left(2^{1 + \Max{f} - \Min{f}} + \expo^2\right)
\end{equation*}
and
\begin{equation*}
\errW \leq \frac{1}{\sqrt{n}}\left(1 - 2 \log(p_{\Omega}\,\epsmooth)\right).
\end{equation*}
\end{theorem}

\subsubsection{EAT channels for Protocol \ref{protocol:rap}}

To apply the EAT to Protocol \ref{protocol:rap}, we need to show that its execution can be described by the composition of $n$ EAT channels. In the entropy accumulation part of our proposed protocol, we have that in each round $i$ an input $X_i$ is sampled from the $\delta$-SV source and, given $X_i$, the preparation box $\prepBox_\energyBound$ applies a local map to its quantum memory $R_{i-1}$ to prepare a state $\rho_i$ satisfying $\Tr[H\rho_i]\leq \omega_{x_i}$. The state $\rho_i$ is then sent to the measurement box $\msrmntBox_i$ which produces the classical value $A_i$ from it. Finally, in step 4 of our protocol, the classical value $C_i$ is produced from $X_i$ and $A_i$. We denote with
$$\C:=\{\mu_{\energyBound^{\rm opt},\delta},-(\mu_{\energyBound^{\rm opt},\delta})^{-1}\}$$ 
the alphabet of the classical registers $C_i$, and let $\P_\C$ be the set of probability distribution over $\C$. We denote the channels evolving the states in our protocol as 
\begin{align*}
\begin{split}
\mathcal{N}_i^{\delta,\energyBound}: R_{i-1} &\rightarrow R_{i}A_iX_iC_i \\
\rho_{R_{i-1}} &\mapsto \rho_{R_iA_iX_iC_i} \,.
\end{split}
\end{align*}

See Fig. \ref{fig:eat-channels} for a graphical depiction of these channels. In the following, we will sometimes omit the channels' dependence on $\delta$ and $\energyBound$ for ease of notation.

\begin{figure}[h]
  \center
  \begin{tikzpicture}[scale=1]
    \draw[rounded corners] (-3.8,-2.2) rectangle (4,1.8);
    \node[] at (-3.3,1.5) {$\mathcal{N}_i^{\delta,\mathbb{\omega}}$};
    \draw[] (-1.5,-1) rectangle (-2.5,1);
    \node[] at (-2,0) {$P_{\mathbb{\omega}}$};
    \draw[] (3.5,-1.5) rectangle (2.5,0.5);
    \node[] at (3,-0.5) {$\mathfrak{D}_i$};
    \draw[] (0.3,-0.1) rectangle (1.3,-2.1);
    \node[] at (0.75,-1){$M_i$};
    %
    %
    %
    %
	\node[] at (-0.3,-0.5){ $\rho_i$};
	\node[] at (-0.6,-1.3){ $\langle H\rangle_{\rho_i}\leq \omega_{x_i}$};

    \draw[->,thick, dashed] (-4,0)node[left] {$R_{i-1}$} to (-2.5,0);
\draw[->,out=0, in=180,thick, dashed] (-1.5,0) to (0.3,-1);
\draw[->,out=0, in=180,thick, dashed] (-1.5,0.35) to (2.5,0.35);
    \draw[->] (1.3,-1) to (2.5,-1);
    \draw[->,thick, dashed] (3.5,0) to (4.5,0) node[right]{$R_i,C_i$};
    \draw[->] (-3,0.8)node[left]{$X_i$} to (-2.5,0.8);
    \draw[->]((-1.5,0.8) -- (4.5,0.8)node[right]{${X}_i$};
    \draw[->]((1.3,-1.8) -- (4.5,-1.8)node[right]{${A}_i$};

  \end{tikzpicture}
  \caption{\raggedright Description of the channel $\mathcal{N}_i^{\delta,\energyBound}$. This channel describes the round $i$ of Protocol~\ref{protocol:rap}'s entropy accumulation part. $X_i$ is sampled from a $\delta$-SV source. Given $X_i$, $\prepBox_\energyBound$ prepares a state $\rho_i$ satisfying the energy constrain $\Tr[H\rho_i]\leq \omega_{x_i}$. Given $\rho_i$, $\msrmntBox_i$ produces the classical value $A_i$. Finally, the classical value $C_i=w(A_i,X_i)$ is produced.}
  \label{fig:eat-channels}
\end{figure}
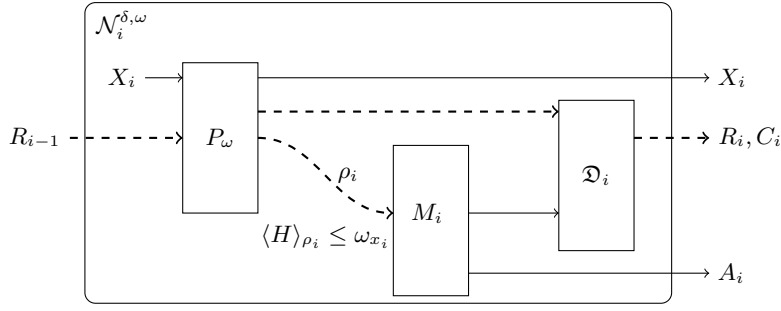

The state after the $n$ rounds of the entropy accumulation part, just before step~\ref{step:abortion} is denoted by
\begin{equation*}
\rho_{\mathbf{A}\mathbf{X}\mathbf{C}E} = \left( \textrm{Tr}_{R_{{n}}} \circ \mathcal{N}_{n} \circ \dots \circ \mathcal{N}_1 \right) \otimes \mathcal{I}_E \rho_{R_0E}
\end{equation*}
In step~\ref{step:abortion} Alice and Bob decide whether to abort the protocol or not. We denote by $\Omega$ the event of not aborting,
\begin{equation*}
\Omega = \Big \{\bar{C} \leq (I_{\mathrm{exp}} + \gamma_{\mathrm{est}}) \Big \} \, .
\end{equation*}
We denote by $\rho_{\mathbf{A}\mathbf{X}\mathbf{C}E|\Omega}$, or short $\rho_{|\Omega}$, the state after the protocol conditioned on not aborting the protocol.

\begin{lemma}
The channels $\mathcal{N}_i$ that evolve the unknown quantum state of Protocol~\ref{protocol:rap2}, are EAT channels, i.e., they satisfy Definition~\ref{def:eat-channels}.
\end{lemma}
\begin{proof}
Condition \ref{item:eat-ch-cond-1} is satisfied because $A_i$ and $X_i$ are classical registers and $C_{i}$ is a classical function of those registers. Condition \ref{item:eat-ch-cond-2} is satisfied because, as is stated in \ref{assumption:markov-chain}, it holds that $I(A_1\dots A_{i-1}:X_i|X_1\dots X_{i-1}E,\lambda) = 0$.\qedhere
\end{proof}

\subsubsection{Min-tradeoff function}

The next step is to give a min-tradeoff function for our EAT channels. To that end, we will retort to the techniques introduced in \cite{8935370}. For fixed $\delta\in [0,1/2)$, let $\energyBound^{\rm opt}=\energyBound^{\rm avg}(\delta)$ (c.f. Eq. \eqref{eq:avg-omega}) and let
\begin{align}\label{eq:sdp-for-min-tradeoff}
h(I)&:=\max_{\{(\tilde{\prob}_{A|X,\lambda'},\energyBound^{\lambda'})\}_{\lambda'\in\{0,1\}^2}}\sum_{a,x} \tilde{p}_{(a,x)}(a|x)&\nonumber\\
&\qquad\qquad\textrm{subject to } &\nonumber\\
&\qquad\qquad\qquad\qquad \sum_{\lambda'}I^\mathrm{lower}_{\energyBound^{\rm opt},\delta}(\prob_{A|X,\lambda'})\leq I &\nonumber\\
&\qquad\qquad\qquad\qquad  \tilde{\prob}_{A|X,\lambda'}\in \tilde{\quantumSet{\energyBound^{\lambda'}}} \text{ for all } \lambda'&\nonumber\\
&\qquad\qquad\qquad\qquad  \sum_{\lambda'}\energyBound^{\lambda'}_x \leq \omega^\mathrm{opt}_x\text{ for all } x &\nonumber\\
&\qquad\qquad\qquad\qquad  \sum_{a,\lambda'} \tilde{p_{\lambda'}}(a|x) = 1.
\end{align}

Notice that from Eq. \eqref{eq:upper-bound-pguess} and Eq. \eqref{eq:sdp-upper} we have that 
$$(\frac{1}{2}+\delta)h(I)=\eta(\delta,\energyBound^{\rm opt},I)\geq p_{{\rm guess}}(A|X,E,\lambda,\delta,I).$$

Just like Eq. \eqref{eq:sdp-upper}, Eq. \eqref{eq:sdp-for-min-tradeoff} is an SDP and it has the following dual (c.f. \cite[Eq. (10)]{8935370}):

\begin{align}\label{eq:sdp-dual}
d_h(I):=& \min_{\alpha,\mathbf{\beta}}\qquad\qquad \alpha\cdot I + \mathbf{\beta}\cdot\energyBound^{\rm opt}\nonumber\\
&\text{subject to }\quad h(I')\leq \alpha\cdot I + \mathbf{\beta}\cdot\energyBound^{\rm opt}\ \ \forall\ I'\in \mathcal{A}
\end{align}
with $\mathcal{A}:=\{I\in\mathbb{R}\mid\exists \prob_{A|X}\in\quantumSet{\energyBound^{\rm opt}} : I^\mathrm{lower}_{\energyBound^{\rm opt},\delta}(\prob_{A|X})=I\}$.

\medskip

Following the techniques in \cite{8935370}, we build min-tradeoff functions from solutions to Eq. \eqref{eq:sdp-dual}. The following lemma is an adaption of \cite[Lemma 3.2]{8935370} to our setting:

\begin{lemma}\label{lemma:min-tradeoff}
For fixed $\delta$, let $(\alpha^\mathrm{opt},\mathbf{\beta}^\mathrm{opt})$ be an optimal solution to $d_h(\mdlIneq{\energyBound^{\rm opt}}^{\rm upper,opt})$ and let $\mathbf{q}^\mathrm{opt}\in\P_\C$ be such that $(\mu_{\energyBound^{\rm opt},\delta},-1/\mu_{\energyBound^{\rm opt},\delta})\cdot \mathbf{q}^\mathrm{opt}=\mdlIneq{\energyBound^{\rm opt}}^{\rm upper,opt}$. Then,
\begin{align*}
f_{\rm min}(\p):=f(\delta)-\frac{(1/2+\delta)\alpha^\mathrm{opt}}{f(\delta)\ln 2}\mathbf{c}\cdot(\p-\mathbf{q}^\mathrm{opt})
\end{align*}
with 
\begin{align*}
\mathbf{c}&=(\mu_{\energyBound^{\rm opt},\delta},-1/\mu_{\energyBound^{\rm opt},\delta})\\
\Min{f_{\rm min}}&=f_{\rm min}(\mathbf{e}_{\mu_{\energyBound^{\rm avg}(\delta),\delta}})\\
\Max{f_{\rm min}}&=f_{\rm min}(\mathbf{e}_{-1/\mu_{\energyBound^{\rm avg}(\delta),\delta}})
\end{align*}
is a min-tradeoff function for the EAT channels $\{\mathcal{N}_i^{\delta,\energyBound^{\rm avg}(\delta)}\}_i$.
\end{lemma}

\begin{proof}
Let $(\alpha^\mathrm{opt},\mathbf{\beta}^\mathrm{opt})$ be an optimal solution to $d_h(\mdlIneq{\energyBound^{\rm opt}}^{\rm upper,opt})$. Notice that,
$$(\frac{1}{2}+\delta)(\alpha^\mathrm{opt}~\mdlIneq{\energyBound^{\rm opt}}^{\rm upper,opt}+\mathbf{\beta}^\mathrm{opt}\cdot \energyBound^{\rm opt})=f(\delta).$$

Let $I\geq\mdlIneq{\energyBound^{\rm opt}}^{\rm upper,opt}$ and let $\mathbf{q}\in\P_\C$ be such that $(\mu_{\energyBound^{\rm opt},\delta},-1/\mu_{\energyBound^{\rm opt},\delta})\cdot \mathbf{q} = I$. Let $\rho:=\{\N_i\otimes I\}(\sigma_{R_{i-1}R'})$ be such that $\rho_{C_i}=\tau_\mathbf{q}$ for some $\sigma_{R_{i-1}R'}$. Then,
\begin{align}\label{eq:non-affine-lower-bound}
H(A_i|X_i,R')_\rho &\geq H_\mathrm{min}(A_i|X_i,R')_\rho\nonumber\\
& = -\log_2 p_\mathrm{guess}(A|X,R',I)\nonumber\\
& \geq -(\frac{1}{2}+\delta) \log_2 d_h(I)\nonumber\\
& \geq -(\frac{1}{2}+\delta)\log_2 (\alpha^\mathrm{opt}~(\mu_{\energyBound^{\rm opt},\delta},-1/\mu_{\energyBound^{\rm opt},\delta})\cdot \mathbf{q}+\mathbf{\beta}^\mathrm{opt}\cdot \energyBound^{\rm opt})\nonumber\\
& =: \tilde{f}_\mathrm{min}(\mathbf{q})
\end{align}
and therefore,
$$\tilde{f}_\mathrm{min}(\mathbf{q}) \leq \inf_{\sigma_{R_{i-1}R'}: \N_i(\sigma)_{\Ci}=\tau_{\mathbf{q}}} H(A_i|X_iR')_{\N_i(\sigma)}.$$

To conclude the proof, in order to have an affine lower-bound, we follow \cite[Lemma 3.2]{8935370} and let $f_\mathrm{min}$ be the first order Taylor expansion of $\tilde{f}_\mathrm{min}$ around a point $\mathbf{q}^\mathrm{opt}$ such that $(\mu_{\energyBound^{\rm opt},\delta},-1/\mu_{\energyBound^{\rm opt},\delta})\cdot \mathbf{q}^\mathrm{opt}=\mdlIneq{\energyBound^{\rm opt}}^{\rm upper,opt}$, that is
\begin{align*}
f_{\rm min}(\p):=f(\delta)-\frac{(1/2+\delta)\alpha^\mathrm{opt}}{f(\delta)\ln 2}(\mu_{\energyBound^{\rm opt},\delta},-1/\mu_{\energyBound^{\rm opt},\delta})\cdot(\p-\mathbf{q}^\mathrm{opt})
\end{align*}
This concludes the proof.
\end{proof}

\subsubsection{Putting all together}
From Lemma \ref{lemma:min-tradeoff} and Theorem \ref{thm:EAT} we find that, either Protocol \ref{protocol:rap} aborts with probability $1-\mathrm{Pr}[\Omega]\geq 1-\epsilon_{\mathrm{EAT}}$ or the lower bound
\begin{align}\label{eq:smooth-min-entropy-lower-bound}
H^{\epsilon_\text{s}}_{\min} \left( \mathbf{A}| \mathbf{X} E \right) > n\cdot f_{\rm min}(I_\mathrm{exp} + \gamma_{\mathrm{est}})-n(\errV(f_{\rm min}) + \errK(f_{\rm min})) - \errW
\end{align}
holds. In Eq. \eqref{eq:smooth-min-entropy-lower-bound}, $I_\mathrm{exp} + \gamma_{\mathrm{est}}$ is a shorthand for any $\p$ such that $(\mu_{\energyBound^{\rm avg}(\delta),\delta},-1/\mu_{\energyBound^{\rm avg}(\delta)),\delta})\cdot\p=I_\mathrm{exp} + \gamma_{\mathrm{est}}$.

The remainder of the proof follows exactly as in the soundness proof of the DI RAP in \cite{kessler2020device}. We state the necessary lemmas (adapted to our scenario and notation) and refer the reader to \cite{kessler2020device} for the proofs.

\medskip

First, we state the definition of a quantum-proof two-source extractor in the Markov model, the extractor used in Protocol \ref{protocol:rap}.

\begin{definition}[\cite{arnon2016quantum}]
\label{def:quantum-two-source-extractor}
	A function $\mathrm{Ext}: \{0,1\}^{n} \times \{0,1\}^{d} \to \{0,1\}^m$ is a $(k_1,k_2,\epsilon)$ quantum-proof two-source extractor in the Markov model, strong in the second source, if for all sources $X_1,X_2$, and quantum side information $C$, where $I(X_1:X_2|C)=0$ and with min-entropy $H_{\text{min}}\left(X_1|C\right)\geq k_1$ and $H_{\text{min}}\left(X_2|C\right)\geq k_2$, we have
\begin{equation*}
\frac{1}{2}\| \rho_{\mathrm{Ext}(X_1,X_2)X_2C} - \rho_{U_m} \otimes \rho_{X_2C} \|_\mathrm{tr} \leq \epsilon \;. 
\end{equation*}
where $\rho_{\mathrm{Ext}(X_1,X_2)C} = \mathrm{Ext} \otimes \mathcal{I}_C \rho_{X_1X_2C}$ and $\rho_{U_m}$ is the fully mixed state on a system of dimension $2^m$. 
\end{definition}

The following lemma states that, with a suitable correction in the security parameter of the extractor, for one of the sources in 
Def. \ref{def:quantum-two-source-extractor} one can replace the lower bound to the conditional min-entropy with a lower bound to the smooth conditional min-entropy.

\begin{lemma}[\cite{kessler2020device}]
\label{lma:smooth-entropy-bound}
	Let $\mathrm{Ext} : \{0,1\}^{n} \times  \{0,1\}^{d} \to \{0,1\}^m$ be a $(k_1,k_2,\varepsilon)$ quantum-proof two-source extractor in the Markov model, strong in the source $X_i$. Then for any Markov state $\rho_{X_1X_2C}$ with $H_{\mathrm{min}}^{\varepsilon_s}(X_1|C)_{\rho} \geq k_1+\log(1/\varepsilon)+1$ and $H_{\mathrm{min}}(X_2|C)_{\rho} \geq k_2+\log(1/\varepsilon)+1$,
$$
\frac{1}{2} \Vert \rho_{\mathrm{Ext}(X_1,X_2)X_iC} - \rho_{U_m} \otimes \rho_{X_iC} \Vert_\mathrm{tr} \leq 6  \left(\varepsilon_s + \varepsilon\right) \;.
$$
\end{lemma}

The following lemma states the secrecy of the output of the extractor when the inputs to it are the output of the entropy accumulation part of Protocol \ref{protocol:rap} and the additional string $\mathbf{Z}$ of length $d$ coming from the SV source.

\begin{lemma}[\cite{kessler2020device}]
\label{lma:RAP-secrecy}
Let $\mathrm{Ext}:\{0,1\}^{2n}\times\{0,1\}^d\rightarrow\{0,1\}^m$ be a $(k_1,k_2,\varepsilon_{ext})$ be a two-source quantum-proof extractor in the Markov model, strong in the second input, such that
\begin{equation} \label{eq:ext_k_values}
\begin{split}
k_1 &\leq n \cdot f_{\rm min}(I_\mathrm{exp} - \gamma_{\mathrm{est}})-n(\errV(f_{\rm min}) + \errK(f_{\rm min})) - \errW(f_{\rm min}) -\log(1/\varepsilon_{ext})-1 \\
k_2 &\leq -d \cdot \log(\frac{1}{2}+\delta) -\log(1/\varepsilon_{ext})-1 
\end{split}
\end{equation}
Consider Protocol~\ref{protocol:rap} using $\mathrm{Ext}$ and any $\varepsilon_{\mathrm{EAT}}, \varepsilon_{\mathrm{s}} \in (0,1)$. 
hen, either the protocol aborts with probability greater than $1 - \varepsilon_{\mathrm{EAT}}$, or for the $m$-bit output $\mathbf{K}$ together with the whole information the adversary possibly has access to, $\Sigma = \mathbf{Z} \mathbf{X} E \lambda$, it holds that
\begin{equation*}
\frac{1}{2} \left\Vert \rho_{\mathbf{K} \Sigma} - \rho_{U_m} \otimes \rho_{\Sigma} \right\Vert_\mathrm{tr} \leq 6 \left(\varepsilon_{\mathrm{s}} +  \varepsilon_{\mathrm{ext}} \right) \,.
\end{equation*}
\end{lemma}

Finally, putting everything together we have the proof of the secrecy part of Theorem \ref{thm:soundness-and-completeness}.

\begin{proof}[Proof of secrecy \cite{kessler2020device}.]
In the following let $\Sigma = \mathbf{Z} \mathbf{X} E \lambda$ be the whole information the adversary has access to.
Starting with Lemma~\ref{lma:RAP-secrecy} we can distinguish two cases.
\begin{enumerate}
\item The protocol aborts with probability greater than $1 - \varepsilon_{\mathrm{EAT}}$.
In that case, we find 
$$
\frac{1}{2}\left( 1 - \mathrm{Pr}[\text{abort}] \right) \left\Vert \rho_{\mathbf{K} \Sigma} - \rho_{U_{m}} \otimes \rho_{\Sigma} \right\Vert_\mathrm{tr} \leq \frac{1}{2} \varepsilon_{\mathrm{EAT}} \left\Vert \rho_{\mathbf{K} \Sigma} - \rho_{U_{m}} \otimes \rho_{\Sigma} \right\Vert_\mathrm{tr} \\
	\leq \varepsilon_{\mathrm{EAT}} \;,
$$
\item The protocol aborts with probability less than $1 - \varepsilon_{\mathrm{EAT}}$.
In that case, using the bound from Lemma~\ref{lma:RAP-secrecy}, we find	
$$
\frac{1}{2}\left( 1 - \mathrm{Pr}[\text{abort}] \right) \left\Vert \rho_{\mathbf{K} \Sigma} - \rho_{U_{m}} \otimes \rho_{\Sigma} \right\Vert_\mathrm{tr} \leq \frac{1}{2}\left\Vert \rho_{\mathbf{K} \Sigma} - \rho_{U_{m}} \otimes \rho_{\Sigma} \right\Vert_\mathrm{tr} \leq  6\left(\varepsilon_{\mathrm{s}} +\varepsilon_{\mathrm{ext}} \right) \,. 
$$
\end{enumerate}
Hence,
$$
\frac{1}{2}(1-\Pr[\mathrm{Abort}])||\rho_{\mathbf{K}\Sigma}-\rho_{U^m}\otimes\rho_{\Sigma}||_\mathrm{tr} \leq 6 \left(\varepsilon_{\mathrm{s}} +\varepsilon_{\mathrm{ext}} \right) + \varepsilon_{\mathrm{EAT}}\;.
$$
\end{proof}

\subsection{Completeness}
\begin{lemma}[Completeness]
\label{lma:completeness}
Let $\SV_\delta$ be any $\delta$-SV source and let $I_\mathrm{exp} \leq \mdlIneq{\energyBound^{\rm avg}}^{\rm upper,opt}$. Then Protocol~\ref{protocol:rap} is complete with completeness parameter $\varepsilon_{\mathrm{c}} \leq \exp \left( - \frac{2 n \mu_{\energyBound,\delta}^2\gamma_{\mathrm{est}}^2}{\left (1+\mu_{\energyBound,\delta}^2 \right )^2} \right)$;
i.e., the probability to abort in an honest implementation is upper bounded by $\varepsilon_{\mathrm{c}}$.
\end{lemma}

\begin{proof}
If we implement our device to perform $n$ independent MDL-like experiments with states and measurements achieving a value $I_{\mathrm{exp}}$ of inequality Eq. \ref{eq:mdl-inequality}, the expectation value of $\bar{C} = \frac{1}{n} \sum_j C_j$ is given by $\mathbb{E}[\bar{C}] = I_{\mathrm{exp}}$. Using Hoeffding's inequality we get the following upper bound on the probability that the protocol aborts:
\begin{align*}
\mathrm{Pr}\left[ \mathrm{aborting} \right] &= \mathrm{Pr}\left[ \bar{C} > (I_\mathrm{exp} + \gamma_{\mathrm{est}}) \right] \\
&= \mathrm{Pr}\left[ \bar{C}-I_\mathrm{exp} > \gamma_{\mathrm{est}} \right] \\
&\leq \exp \left( - \frac{2 n \mu_{\energyBound,\delta}^2\gamma_{\mathrm{est}}^2}{\left (1+\mu_{\energyBound,\delta}^2 \right )^2} \right) \, . \qedhere
\end{align*}
\end{proof}

\end{widetext}

\begin{thebibliography}{25}
\providecommand{\natexlab}[1]{#1}
\providecommand{\url}[1]{\texttt{#1}}
\expandafter\ifx\csname urlstyle\endcsname\relax
  \providecommand{\doi}[1]{doi: #1}\else
  \providecommand{\doi}{doi: \begingroup \urlstyle{rm}\Url}\fi

\bibitem[Santha and Vazirani(1986)]{santha1986generating}
Miklos Santha and Umesh~V Vazirani.
\newblock Generating quasi-random sequences from semi-random sources.
\newblock \emph{Journal of computer and system sciences}, 33\penalty0
  (1):\penalty0 75--87, 1986.

\bibitem[Colbeck and Renner(2012)]{colbeck2012free}
Roger Colbeck and Renato Renner.
\newblock Free randomness can be amplified.
\newblock \emph{Nature Physics}, 8\penalty0 (6):\penalty0 450--453, 2012.

\bibitem[Gallego et~al.(2013)Gallego, Masanes, De~La~Torre, Dhara, Aolita, and
  Ac{\'\i}n]{gallego2013full}
Rodrigo Gallego, Lluis Masanes, Gonzalo De~La~Torre, Chirag Dhara, Leandro
  Aolita, and Antonio Ac{\'\i}n.
\newblock Full randomness from arbitrarily deterministic events.
\newblock \emph{Nature Communications}, 4\penalty0 (1):\penalty0 1--7, 2013.

\bibitem[Brand{\~a}o et~al.(2016)Brand{\~a}o, Ramanathan, Grudka, Horodecki,
  Horodecki, Horodecki, Szarek, and Wojew{\'o}dka]{brandao2016realistic}
Fernando~GSL Brand{\~a}o, Ravishankar Ramanathan, Andrzej Grudka, Karol
  Horodecki, Micha{\l} Horodecki, Pawe{\l} Horodecki, Tomasz Szarek, and Hanna
  Wojew{\'o}dka.
\newblock Realistic noise-tolerant randomness amplification using finite number
  of devices.
\newblock \emph{Nature communications}, 7\penalty0 (1):\penalty0 1--6, 2016.

\bibitem[Kessler and Arnon-Friedman(2020)]{kessler2020device}
Max Kessler and Rotem Arnon-Friedman.
\newblock Device-independent randomness amplification and privatization.
\newblock \emph{IEEE Journal on Selected Areas in Information Theory},
  1\penalty0 (2):\penalty0 568--584, 2020.

\bibitem[Foreman et~al.(2020)Foreman, Wright, Edgington, Berta, and
  Curchod]{foreman2020practical}
Cameron Foreman, Sherilyn Wright, Alec Edgington, Mario Berta, and Florian~J
  Curchod.
\newblock Practical randomness and privacy amplification.
\newblock \emph{arXiv preprint arXiv:2009.06551}, 2020.

\bibitem[Gerhardt et~al.(2011)Gerhardt, Liu, Lamas-Linares, Skaar, Kurtsiefer,
  and Makarov]{gerhardt2011full}
Ilja Gerhardt, Qin Liu, Antia Lamas-Linares, Johannes Skaar, Christian
  Kurtsiefer, and Vadim Makarov.
\newblock Full-field implementation of a perfect eavesdropper on a quantum
  cryptography system.
\newblock \emph{Nature communications}, 2\penalty0 (1):\penalty0 1--6, 2011.

\bibitem[Paw{\l}owski and Brunner(2011)]{pawlowski2011semi}
Marcin Paw{\l}owski and Nicolas Brunner.
\newblock Semi-device-independent security of one-way quantum key distribution.
\newblock \emph{Physical Review A}, 84\penalty0 (1):\penalty0 010302, 2011.

\bibitem[Zhou et~al.(2015)Zhou, Li, Wang, Li, Gao, and Wen]{zhou2015semi}
Yu-Qian Zhou, Hong-Wei Li, Yu-Kun Wang, Dan-Dan Li, Fei Gao, and Qiao-Yan Wen.
\newblock Semi-device-independent randomness expansion with partially free
  random sources.
\newblock \emph{Physical Review A}, 92\penalty0 (2):\penalty0 022331, 2015.

\bibitem[Brask et~al.(2017)Brask, Martin, Esposito, Houlmann, Bowles, Zbinden,
  and Brunner]{brask2017megahertz}
Jonatan~Bohr Brask, Anthony Martin, William Esposito, Raphael Houlmann, Joseph
  Bowles, Hugo Zbinden, and Nicolas Brunner.
\newblock Megahertz-rate semi-device-independent quantum random number
  generators based on unambiguous state discrimination.
\newblock \emph{Physical Review Applied}, 7\penalty0 (5):\penalty0 054018,
  2017.

\bibitem[Van~Himbeeck et~al.(2017)Van~Himbeeck, Woodhead, Cerf,
  Garc{\'\i}a-Patr{\'o}n, and Pironio]{van2017semi}
Thomas Van~Himbeeck, Erik Woodhead, Nicolas~J Cerf, Ra{\'u}l
  Garc{\'\i}a-Patr{\'o}n, and Stefano Pironio.
\newblock Semi-device-independent framework based on natural physical
  assumptions.
\newblock \emph{Quantum}, 1:\penalty0 33, 2017.

\bibitem[Miklin et~al.(2020)Miklin, Borka{\l}a, and
  Paw{\l}owski]{miklin2020semi}
Nikolai Miklin, Jakub~J Borka{\l}a, and Marcin Paw{\l}owski.
\newblock Semi-device-independent self-testing of unsharp measurements.
\newblock \emph{Physical Review Research}, 2\penalty0 (3):\penalty0 033014,
  2020.

\bibitem[Tavakoli(2020)]{tavakoli2020semi}
Armin Tavakoli.
\newblock Semi-device-independent certification of independent quantum state
  and measurement devices.
\newblock \emph{Physical Review Letters}, 125\penalty0 (15):\penalty0 150503,
  2020.

\bibitem[Note1()]{Note1}
Note1.
\newblock In this context, we say that a source $\protect \mathcal {S}$ is
  \protect \emph {public} if, \protect \textbf {after} manufacturing the
  protocol's devices, the adversary can have access to the bits produced by
  $\protect \mathcal {S}$ (i.e., the inputs to the device).

\bibitem[Zhou et~al.(2016)Zhou, Gao, Li, Li, and Wen]{PhysRevA.94.032318}
Yu-Qian Zhou, Fei Gao, Dan-Dan Li, Xin-Hui Li, and Qiao-Yan Wen.
\newblock Semi-device-independent randomness expansion with partially free
  random sources using $3\ensuremath{\rightarrow}1$ quantum random access code.
\newblock \emph{Phys. Rev. A}, 94:\penalty0 032318, Sep 2016.
\newblock \doi{10.1103/PhysRevA.94.032318}.
\newblock URL \url{https://link.aps.org/doi/10.1103/PhysRevA.94.032318}.

\bibitem[Arnon-Friedman et~al.(2016)Arnon-Friedman, Portmann, and
  Scholz]{arnon2016quantum}
Rotem Arnon-Friedman, Christopher Portmann, and Volkher~B Scholz.
\newblock Quantum-proof multi-source randomness extractors in the markov model.
\newblock In \emph{11th Conference on the Theory of Quantum Computation,
  Communication and Cryptography (TQC 2016)}. Schloss Dagstuhl-Leibniz-Zentrum
  fuer Informatik, 2016.

\bibitem[Thinh et~al.(2013)Thinh, Sheridan, and Scarani]{sheridan2013bell}
Le~Phuc Thinh, Lana Sheridan, and Valerio Scarani.
\newblock Bell tests with min-entropy sources.
\newblock \emph{Physical Review A}, 87\penalty0 (6):\penalty0 062121, 2013.

\bibitem[P{\"u}tz et~al.(2014)P{\"u}tz, Rosset, Barnea, Liang, and
  Gisin]{putz2014arbitrarily}
Gilles P{\"u}tz, Denis Rosset, Tomer~Jack Barnea, Yeong-Cherng Liang, and
  Nicolas Gisin.
\newblock Arbitrarily small amount of measurement independence is sufficient to
  manifest quantum nonlocality.
\newblock \emph{Physical review letters}, 113\penalty0 (19):\penalty0 190402,
  2014.

\bibitem[P{\"u}tz and Gisin(2016)]{putz2016measurement}
Gilles P{\"u}tz and Nicolas Gisin.
\newblock Measurement dependent locality.
\newblock \emph{New journal of Physics}, 18\penalty0 (5):\penalty0 055006,
  2016.

\bibitem[Bancal et~al.(2014)Bancal, Sheridan, and Scarani]{bancal2014more}
Jean-Daniel Bancal, Lana Sheridan, and Valerio Scarani.
\newblock More randomness from the same data.
\newblock \emph{New Journal of Physics}, 16\penalty0 (3):\penalty0 033011,
  2014.

\bibitem[Van~Himbeeck and Pironio(2019)]{van2019correlations}
Thomas Van~Himbeeck and Stefano Pironio.
\newblock Correlations and randomness generation based on energy constraints.
\newblock \emph{arXiv preprint arXiv:1905.09117}, 2019.

\bibitem[Dupuis et~al.(2020)Dupuis, Fawzi, and Renner]{dupuis2016entropy}
Frederic Dupuis, Omar Fawzi, and Renato Renner.
\newblock Entropy accumulation.
\newblock \emph{Communications in Mathematical Physics}, 379:\penalty0
  867--913, 2020.

\bibitem[Tavakoli(2021)]{tavakoli2021semi}
Armin Tavakoli.
\newblock Semi-device-independent framework based on restricted distrust in
  prepare-and-measure experiments.
\newblock \emph{arXiv preprint arXiv:2101.07830}, 2021.

\bibitem[Arnon-Friedman et~al.(2019)Arnon-Friedman, Renner, and
  Vidick]{arnon2019simple}
Rotem Arnon-Friedman, Renato Renner, and Thomas Vidick.
\newblock Simple and tight device-independent security proofs.
\newblock \emph{SIAM Journal on Computing}, 48\penalty0 (1):\penalty0 181--225,
  2019.

\bibitem[{Brown} et~al.(2020){Brown}, {Ragy}, and {Colbeck}]{8935370}
P.~J. {Brown}, S.~{Ragy}, and R.~{Colbeck}.
\newblock A framework for quantum-secure device-independent randomness
  expansion.
\newblock \emph{IEEE Transactions on Information Theory}, 66\penalty0
  (5):\penalty0 2964--2987, 2020.
\newblock \doi{10.1109/TIT.2019.2960252}.

\end{thebibliography}
\end{document}